\documentclass[12pt]{amsart}
\usepackage[utf8]{inputenc}
\usepackage[T1]{fontenc}
\usepackage[T2A]{fontenc}
\usepackage[english]{babel}
\usepackage[margin=1.2in]{geometry}
\usepackage{float}
\usepackage{xurl}
\usepackage{xfrac}
\usepackage{mathrsfs}
\usepackage{dsfont}
\usepackage{amsmath}
\usepackage{amsthm}
\usepackage{xcolor}
\usepackage[justification=centering]{caption}
\usepackage{subcaption}
\usepackage{stmaryrd}
\usepackage{amssymb}
\usepackage{hyperref}
\usepackage{booktabs}
\usepackage{mathtools, cases}
\usepackage{hyperref}
\usepackage{multicol}
\usepackage[bibstyle=ieee, citestyle=numeric, backend=biber, sortcites=false, sorting=nyt, maxnames=5]{biblatex}
\DeclareNameAlias{default}{family-given}

\bibliography{biblio}

\newtheorem{theorem}{Theorem}
\newtheorem{lemma}[theorem]{Lemma}
\newtheorem{proposition}{Proposition} 

\theoremstyle{definition}
\newtheorem*{assumption-definition}{Assumption/Definition}
\newtheorem{definition}{Definition}
\newtheorem{remark}{Remark}
\newtheorem{example}{Example}

\newcommand{\bbR}{\mathbb{R}}

\AtBeginEnvironment{example}{%
  \pushQED{\qed}%
}
\AtEndEnvironment{example}{\popQED\endexample}

\begin{document}

\title{Attraction of the core and the cohesion flow}

\author{Dylan Laplace Mermoud}
\address{Unit{\'e} de Math{\'e}matiques Appliqu{\'e}es, ENSTA, Institut Polytechnique de Paris, 91120 Palaiseau, France, CEDRIC, Conservatoire National des Arts et M{\'e}tiers, 75003, Paris}
\email{dylan.laplace.mermoud@gmail.com}
\thanks{The author wishes to thank the anonymous referees and editor for their valuable comments that greatly improved the presentations of the present results.} 

\date{\today}

\subjclass[2020]{MSC Primary 91A12; Secondary 91B24}
\keywords{Core, Stability, Cooperative Game Theory, Dynamical Systems}

\begin{abstract}
We adopt a continuous-time dynamical system approach to study the evolution of the state of a game driven by the willingness to reduce the total dissatisfaction of the coalitions about their payment. Inspired by the work of Grabisch and Sudh{\"o}lter about core stability, we define a vector field on the set of preimputations from which is defined, for any preimputation, a \emph{cohesion curve} describing the evolution of the state. We prove that for each preimputation, there exists a unique cohesion curve. Subsequently, we show that, for the cohesion flow of a balanced game, the core is the unique minimal attractor of the flow, the realm of which is the whole preimputation set. These results improve our understanding of the ubiquity of the core in the study of cooperative games with transferable utility. 

\end{abstract}

\maketitle

\section{Introduction}

One of the key solution concepts in cooperative game theory is the core. We can loosely define the core as the set of states of the game that satisfy every coalition. Thanks to its intuitive definition, the popularity of the core never decreased since its introduction as a solution concept by Shapley~\cite{shapley1953theory, shubik1992game, zhao2018three} and its study by Gillies~\cite{gillies1959solutions}. 

\medskip

Moreover, many single-valued solution concepts, such as the (pre)nucleolus or the Shapley value, are studied comparatively to the core. Indeed, one of the main qualities of the nucleolus is to always belong to the core when it is non-empty, and at a central position. For the Shapley value, it is said to be \emph{stable} when it belongs to the core, as it happens for convex games, making it very appealing as a one-point solution concept. 

\medskip 

Even though the core is very well appreciated for being the set of solutions from which deviations are not reasonable, relatively little is known about its capacity to attract states and to make them evolve towards it, except for one specific type of discrete dynamics. Such a dynamic approach in cooperative game theory was very popular in its early days, with the concept of stable sets, introduced by von Neumann and Morgenstern in the same work in which they define cooperative games~\cite{von1944theory}. Roughly speaking, a stable set is defined as a set of states such that no two states from this set are comparable, and any state outside the set can be improved into a state of the stable set, via a mechanism called domination. 

\medskip 

B{\'e}al, R{\'e}mila and Solal~\cite{beal2013optimal} have shown that at most \(n-1\) successive dominations were enough to reach the core of an \(n\)-player game, and Grabisch and Sudh{\"o}lter~\cite{grabisch2021characterization} characterized the games for which a unique improvement is sufficient. Stearns~\cite{stearns1968convergent} and Cesco~\cite{cesco1998convergent} have also studied discrete-time dynamics on cooperative games. Billera and Wu~\cite{billera1977dynamic} have developed a continuous-time dynamical system similar to ours, but relative to the kernel, another solution concept of cooperative game theory. 

\medskip 

To pursue the investigation of the attraction of the core, and to extend it to the continuous-time setting, we introduce an intuitive dynamical system. Its vector field is defined as a weighted sum of vectors, each of them representing the direction along which the payment of dissatisfied coalitions increases the most. Each of these vectors is weighted by the intensity of the dissatisfaction of the corresponding coalition.

\subsection{Contributions}

In this paper, we introduce two maps defined on the set of states of the game. The first one is a scalar field \(\vartheta\) called the \emph{dissatisfaction field} which associates to each state a value measuring its instability. The second one is a vector field \(\varphi\) called the \emph{cohesion field} which aggregates the transformations of the current state wished by the frustrated coalitions. These two maps are related by the formula: 
\[
\varphi = - \nabla \vartheta.
\]
Moreover, when the core of the game is non-empty, it is precisely the set of states on which the two aforementioned fields vanish. 

\medskip 

The second and most important contribution of the paper is the proof that the core is the unique minimal global attractor of the newly defined dynamical system. Being an attractor means that, for a given initial state of the game, if it evolves by following the cohesion flow, i.e., by following the aggregate direction wanted by the frustrated coalitions, it eventually reaches the core. The core is a global attractor, so it attracts every possible state of the game, no matter how far they initially are from the core. Moreover, it is the unique minimal one, hence it is the only actual attractor, every other attractors are so only because they contain the core. 

\subsection{Organization of the paper}

First, we recall some important concepts of cooperative game theory and dynamical systems in Section~\ref{sec: preliminaries}, such as the definition of the core and of an attractor. In Section~\ref{sec: dissatisfaction} we introduce the \emph{dissatisfaction field}, the map associating to each state of the game its level of instability. In this section, we also show that the dissatisfaction field is compatible with an important concept of domination-based dynamics, called the \emph{outvoting} relation. Section~\ref{sec: cohesion} contains the most important contributions of the paper, including the definition of the \emph{cohesion field}, its connection with the dissatisfaction field, and the exposition of some of the properties of the dynamical system it induces on the set of states of the game. We conclude with Section~\ref{sec: conclusion}, which also comprises perspectives of future works.


\section{Preliminaries}\label{sec: preliminaries}

\subsection{Cooperative games and their cores}

We start with the definition of cooperative games with transferable utility. 

\begin{definition}[\textcite{von1944theory}]\leavevmode \newline 
A \emph{cooperative game with transferable utility} is an ordered pair \((N, v)\) with 
\begin{itemize}
\item \(N\) is a finite nonempty set of \emph{players}; 
\item \(v: 2^N \setminus \{\emptyset\} \to \mathbb{R}\) is a set function called the \emph{coalition function}. 
\end{itemize}
The nonempty subsets of \(N\) are called \emph{coalitions}, and their set is denoted by \(\mathcal{N}\). 
\end{definition}

Without any loss of generality, we can assume that \((N, v)\) satisfies \(v(N) = 0\). Any game can be transformed into a game of this form, by defining \( v'(S) = v(S) - \frac{s}{n} v(N)\) for any coalition \(S \in \mathcal{N}\), with \(s = \lvert S \rvert\) and \(n = \lvert N \rvert\). The core of the new game can therefore be deduced from the original one by translation, and allows for a simpler manipulation of the geometry of the game. 

\medskip 

The payoff vectors that allocate the value of the grand coalition among the players are called \emph{premputations}, and their set is defined by 
\[
X \coloneqq \{x \in \mathbb{R}^N \mid x(N) = 0\},  
\]
with \(x(S) = \sum_{i \in S} x_i\). In the usual setting where \(v(\emptyset) = 0\) instead of \(v(N) = 0\), these vectors are called \emph{side payments}. According to the reader's preferences, \(X\) can be seen as the space of preimputations for \(v(N) = 0\), or of side payments for \(v(\emptyset) = 0\). 

\begin{remark}
The space \(X\) is a vector subspace of the whole allocation vector space \(\mathbb{R}^N\) of dimension \(n - 1\). Therefore, any sum of preimputations is a preimputation, and any preimputation multiplied by a real number is a preimputation, which is not the case when \(v(N) \neq 0\). 
\end{remark}

We denote by \(\mathbf{1}^S\) the indicator vector of \(S\) in \(N\), i.e., \(\mathbf{1}^S_i = 1\) if \(i \in S\), and \( \mathbf{1}^S_i = 0\) otherwise. We define the preimputation \(\eta^S\) as the orthogonal projection of \(\mathbf{1}^S\) onto the set \(X\), i.e., \(\eta^S = \mathbf{1}^S - \frac{s}{n} \mathbf{1}^N\), or, equivalently, 
\[
\eta^S_i = \begin{cases}
- \frac{s}{n}, & \quad \text{if } i \not \in S, \\
1 - \frac{s}{n}, & \quad \text{if } i \in S, 
\end{cases}
\]
Among the preimputations with the same norm, it is the unique one which maximises the payment of \(S\). Note that, for all preimputations \(x \in X\), we have \(x(S) = \langle x, \eta^S \rangle\), with \(\langle \cdot, \cdot \rangle\) denoting the standard scalar product. 

\medskip 

For any coalition \(S\), we define the map \(e_S: X \to \mathbb{R}\), called the \emph{excess function}, by
\[
e_S: x \mapsto e_S(x) \coloneqq v(S) - x(S). 
\] 
The excess of \(S\) at \(x\), denoted by \(e_S(x)\) represents the amount of additional utility that the coalition can acquire by defecting and working on its own. We denote by \(C(v)\) the set of preimputations for which the excess of any coalition is nonpositive, i.e., 
\[
C(v) \coloneqq \{x \in X \mid e_S(v) \leq 0, \forall S \in \mathcal{N}\}. 
\]
The \(C(v)\) is called the \emph{core} of the game. When the players of the game reach an agreement that belongs to the core, no coalition has an interest in defecting. Hence, the cooperation between them can be expected to last. For a given preimputation \(x \in X\), we denote by \(\phi(x) \subseteq \mathcal{N}\) the set of coalitions having a positive excess at \(x\), i.e., 
\[
\phi(x) \coloneqq \{ S \in \mathcal{N} \mid e_S(x) > 0\}. 
\]
A collection of coalitions is said to be \emph{feasible} (Grabisch and Sudh{\"olter}~\cite{grabisch2021characterization}) if there exists a preimputation \(x \in X\) such that \(S = \phi(x)\). The map \(\phi\) defines a equivalence relation \(\sim_\phi\) on the set of preimputations, defined by \(x \sim_\phi y \) if \(\phi(x) = \phi(y)\). The equivalence classes are called the \emph{regions} of \(X\), and there is only a finite number of them. The region associated with the empty collection of coalitions is the core.

\subsection{Differential geometry and dynamical systems}

The goal of this paper is to define a \emph{dynamical system} on the set \(X\) of preimputations. Roughly speaking, a dynamical system is represented by a map \(\Phi: I \times X \to X\), with \(I\) an interval \(I \subseteq \mathbb{R}\) representing the time during which the dynamical systems occurs. In our study, the map \(\Phi\) is defined via an \emph{initial value problem} composed of an ordinary differential equation and an initial condition: 
\[
\begin{cases}
\partial_t \Phi(t, x) & = f(\Phi(t, x)), \\
\Phi(0, x) & = x_0, 
\end{cases}
\]
where \(f: X \to X\) is a \emph{vector field}, associating with each preimputation \(x\) its velocity \(f(x)\), and \(x_0\) is the initial state of the game. 

\medskip 

The gradient of a function \(g : X \to \mathbb{R}\) at a given preimputation \(x \in X\), denoted by \(\nabla_x f\), is defined by 
\[
\langle \nabla_x g, h \rangle = \lim_{t \to 0} \frac{g(x + th) - g(x)}{t}, 
\]
with \(\langle \cdot, \cdot \rangle\) denoting the usual inner product in \(\mathbb{R}^n\). The gradient \(\nabla_x g\) represents the direction along which the function \(g\) has the steepest increase around \(x\). 

\medskip 

If the vector field \(f\) used in the definition of our dynamical system is itself the gradient of a scalar field \(g\), i.e., \( f = \nabla g \), we call \(g\) the \emph{potential} of \(f\). 

\medskip 

Let \(x\) be the initial state of the game. The map \(\gamma_x: \Phi( \cdot, x) : I \to X\) is called the \emph{flow} through \(x\). Using a slight abuse of notation, the set \(\gamma_x = \{\Phi(t, x) \mid t \in I\}\) is called the \emph{flow curve} through \(x\). If \(\Phi\) is determined as the solution of an initial value problem involving a function \(f\) having a potential \(- g\), then the flow \(\gamma_x\) represents the trajectory of \(x\) such that, at any time, the direction taken reduces as much as possible the value of \(g\). 

\medskip 

From all this data, it is possible to study where these flows lead, for all the possible initial conditions. For any imputation \(x \in X\), the set \(\gamma_x^+ = \bigcup_{t \geq 0} \gamma_x(t)\) is called the \emph{forward} semi-orbit of \(x\). The \emph{omega limit set} \(\Lambda^+(x)\) of the orbit \(\gamma_x\) is defined by 
\[
\Lambda^+(x) = \bigcap_{t \in \mathbb{R}} \gamma^+_{\gamma_x(t)}. 
\]
For \(Y \subseteq X\) and \(x \in X\), we write 
\[
\rho(Y, x) = \inf_{y \in Y} \lVert x - y \rVert, \qquad \text{and} \qquad B_r(Y) = \{x \in X \mid \rho(Y, x) < r\}. 
\]

\begin{definition}[Auslander, Bhatia, and Seibert~\cite{auslander1964attractors}]
Let \(Y\) be a nonempty compact subset of \(X\). The set \(Y\) is said to be an \emph{(asymptotically Lyapunov stable) attractor} if
\begin{enumerate}
\item\label{item: invariant} for all \(y \in Y\), we have \(\gamma_y \subseteq Y\), 
\item for all \(\varepsilon > 0\), there exists \(\delta > 0\) such that \(x \in B_\delta(Y)\) implies \(\gamma_x^+ \subseteq B_\varepsilon(Y)\), 
\item\label{item: omega-set} for some \(\delta > 0\), \(x \in B_\delta(Y)\) implies that \(\Lambda^+(x)\) is an nonempty subset of \(Y\). 
\end{enumerate}
A set satisfying the condition~\eqref{item: invariant} is said to be \emph{invariant}. 
\end{definition}

In the literature, the definitions of attractors may vary (see Milnor~\cite{milnor1985concept}), however, the one presented above is a common one. For a set \(Y\) to be an attractor, it first requires that the set is invariant, then that each preimputation close enough from \(Y\) evolves continuously in the neighbourhood of \(Y\), and finally that, eventually, the preimputation evolves into an element of \(Y\) and never leave it afterwards. 

\medskip 

The \emph{realm of attraction} of a set \(Y\) is the union of all orbits with the property that their omega limit sets are nonempty and contained in \(Y\) (see condition~\eqref{item: omega-set}). An attractor is \emph{global} if its realm of attraction is the whole domain of \(\Phi(t, \cdot)\).


\section{The dissatisfaction field as a measurement of instability}\label{sec: dissatisfaction}

In this paper, we are interested in preimputations that lie outside the core, that is, preimputations for which there exist coalitions that are dissatisfied with their assigned payments. These coalitions contribute to the total value \(v(N)\) created by the grand coalition, but receive less than what they could achieve on their own. By balancedness, if any coalition \(S \in \mathcal{N}\) leaves the grand coalition because its payment at \(x \in X\) is not sufficient, what is left to the complement of this coalition is less than its current payment \(x(N \setminus S)\). As a result, \(S\) has an incentive to leave or demand a higher payment, while the rest of the players want \(S\) to stay without compensating it further. This exhibits some instability of the game at the state \(x\) outside the core. In this section, we define a map on the set of preimputations that estimates this instability. 

\medskip 

The excess \(e_S(x)\) of a coalition \(S\) at a preimputation \(x\) represents how much utility the coalition is missing in order to be satisfied with the current state of the game. To analyze the simultaneous dissatisfaction of all coalitions, we unfold the preimputation \(x\) into a new vector \(z \coloneqq \left( z_S \right)_{S \in \mathcal{N}} \in \mathbb{R}^\mathcal{N}\), where each element of the basis of \(\mathbb{R}^\mathcal{N}\) corresponds to a coalition. Then, the deviation between \(z\) and the game seen as a vector in the space \(\mathbb{R}^\mathcal{N}\)  is given by \(\sum_{S \in \mathcal{N}} \left( v_S - z_S \right)^2 = \sum_{S \in \mathcal{N}} e_S(x)^2\). This leads us to define a map from \(X\) to \(\mathbb{R}\), which measures the total dissatisfaction at a given preimputation. We denote this function by \(\vartheta\), and define it as: 
\[
\begin{tabular}{lccl}
    \( \vartheta: \) & \( X \) & \( \longrightarrow \) & \( \bbR \) \\[0.1cm]
    & \( x \) & \( \longmapsto \) & \( \displaystyle \vartheta(x) \coloneqq \frac{1}{2} \sum_{S \in \phi(x)} e_S(x)^2 = \frac{1}{2} \sum_{S \in \mathcal{N}} (e_S(x)^+)^2 \)
  \end{tabular} 
\]
where \(t^+ \coloneqq \max \{0, t\}\). The function \(\vartheta\) maps each preimputation into a real number, hence is a \emph{scalar field}. Unlike the previous example of the distance between the game and the preimputation, we have here a one-half coefficient for convenience and the sum only contains terms corresponding to coalitions for which \(x\) is unsatisfactory. Therefore, the number \(\vartheta(x)\) represents the cumulative dissatisfaction at the state \(x\). 

\begin{example}
Let \((N, v)\) be the game defined on \(N = \{a, b, c\}\) by \(v(S) = - \frac{3}{2} \lvert S \rvert \lvert N \setminus S \rvert\). 

\begin{figure}[H]
\begin{center}
\includegraphics{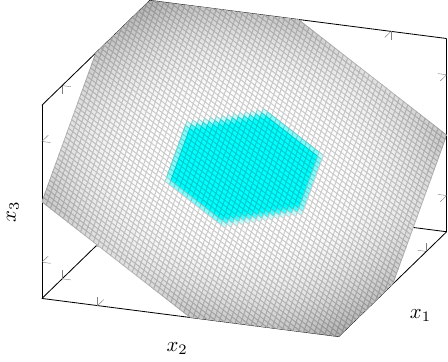}
\caption{Dissatisfaction field for the game \(v(S) = -\frac{3}{2} \lvert S \rvert \lvert N \setminus S \rvert\). The blue preimputations form the core. The black and white preimputations are outside the core, and the darkest are the ones having a larger cumulative dissatisfaction.}
\end{center}
\end{figure}

Let \(x = (-4, -3, 7)\) be a preimputation. The set of aggrieved coalitions \(\phi(x)\) is composed of \(\{a\}\) and \(\{a, b\}\), because 
\[
e_{\{a\}}(x) = -3 - (-4) = 1 > 0, \qquad \text{and} \qquad e_{\{a, b\}}(x) = -3 - (-7) = 4 > 0. 
\]
The dissatisfaction \(\vartheta (x) \) of the preimputation \(x\) is given by \(\vartheta(x) = \frac{1}{2} (1^2 + 4^2) = 8.5\). 
\end{example}

To highlight the relevance of this dissatisfaction field for the study of the dynamics of cooperative games and core stability, we demonstrate that the dissatisfaction field is compatible with a transitive refinement of the von Neumann-Morgenstern domination relation~\cite{von1944theory}, called \emph{outvoting}~\cite{grabisch2021characterization}. A preimputation \(x\) outvotes another preimputation \(y\), written \(x \succ y\), if there exists a coalition \(S \in \mathcal{N}\) such that 
\begin{itemize}
\item \(x\) dominates \(y\) via \(S\), i.e., \(e_S(x) \geq 0\) and \(x_i > y_i\) for all \(i \in S\), and
\item \(x(T) \geq v(T)\), for all coalitions \(T\) not contained in \(S\). 
\end{itemize}

If the current state \(y\) of the game is outvoted by a preimputation \(x\) via \(S\), the players in \(S\) will demand to evolve the state towards \(x\) by threatening to defect. Moreover, no players outside \(S\) can form a coalition with a positive excess at \(x\), i.e., there are no reasonable objections to \(x\). Naturally, the state \(x\) shows more stability than \(y\). 

\begin{proposition}
	Let \(x, y\) be two preimputations such that \(x \succ y\). Then \(\vartheta(x) \leq \vartheta(y)\).
\end{proposition}

\begin{proof}
First, we prove that there are no new summands in the sum defining \(\vartheta(x)\). Let denote by \(S\) a coalition such that \(x \succ_S y\). Because \(x\) is outvoting, it satisfies \(x(T) \geq v(T)\) for all coalitions \(T\) not contained in \(S\). So, they cannot be new coalitions to sum over that are not contained in \(S\), as \(\phi(x) \subseteq S\). Moreover, for each player \(i \in S\), we have that \(x_i > y_i\), hence \(x(T) \geq y(T)\) for every subcoalition \(T \subseteq S\), which implies that the excesses at \(x\) and \(y\) satisfy \(0 \leq e_T(x) < e_T(y)\), with the non-negativity coming from the affordability of the domination relation. So, \(\phi(x) \subseteq \phi(y)\). Because all the excesses at \(x\) are lower than at \(y\), we have \(\vartheta(x) \leq \vartheta(y)\). 
\end{proof}

Because we sum over several coalitions, it can seem like we are taking into account several times the discontentment of the same player if she belongs to multiple coalitions. However, in our context, it is the coalitions that are responsible for the dynamics, hence only the discontentment of the coalitions is relevant. Indeed, each aggrieved coalition can impose a new preimputation via domination or outvoting, and their motivation to do so increases with their excess. Therefore, \(\vartheta(x)\) measures the instability of the state \(x\): the larger it gets, the more coalitions are willing to renegotiate or to defect. 

\medskip 

Moreover, when the core is nonempty, it coincides with the set of minimizers of \(\vartheta\). 

\begin{proposition}\label{prop: alt-def-core}
	For any game \((N, v)\), we have \(C(v) = \{x \in X \mid \vartheta(x) = 0\}\). 
\end{proposition}

\begin{proof}
	Let \(x \in C(v)\). We have \(\phi(x) = \emptyset\), then \(\vartheta(x) = 0\). Let \(y \in X \setminus C(v)\). There exists a coalition \(S\) such that \(e_S(y) > 0\). Then, \(\vartheta(y) \geq \frac{1}{2} e_S(y)^2 > 0\). 
\end{proof}

The property of the core reflected in Proposition~\ref{prop: alt-def-core} is probably the reason of its popularity, and is equivalent to coalitional rationality. It is commonly interpreted that, if the players agree on a core preimputation, they stay committed to it. In this sense, numerous author call core allocations \emph{stable}, especially in the Shapley value literature. This behaviour was summarized by Shubik~\cite{shubik1982game} as follows: ``a game that has a core has less potential for social conflict than one without a core''.

\medskip 

Nonetheless, the power of attraction of the core is rarely discussed. When the state of the game initially belongs to the core, we expect it not to change. However, what happens when the initial state lies outside the core is less studied, at least for continuous-time dynamics. For discrete time steps, this study is often performed by analyzing the interaction between the core and the von Neumann-Morgenstern stable sets. One major result in this literature is the characterization of the games for which the core is a stable set by Grabisch and Sudh{\"o}lter~\cite{grabisch2021characterization}.


\section{The cohesion field and its flow}\label{sec: cohesion}

One difficulty for the core to attract the state of the game is that the different frustrated coalitions may not have enough common interest to find a direction along which the state should evolve to simultaneously reduce all excesses. Hence, each coalition tries to push the state into contradictory directions, leading to a status quo. The following definition introduces a tool to capture and summarize all these considerations. 

\begin{definition}
The \emph{cohesion field} of the game \((N, v)\) is defined by 
\[
\begin{tabular}{lccl}
    \( \varphi: \) & \( X \) & \( \longrightarrow \) & \( X \) \\[0.1cm]
    & \( x \) & \( \longmapsto \) & \( \displaystyle \varphi(x) \coloneqq \sum_{S \in \phi(x)} e_S(x) \eta_S \).
  \end{tabular} 
\]
\end{definition}

The cohesion field associates to any preimputation an \(n\)-dimensional preimputation, hence is a \emph{vector field}. This vector is the aggregation of all the directions that the frustrated coalitions wish to take.

\begin{example}[Continued]
Let \((N, v)\) be the game of the previous example, with 
\[
N = \{a, b, c\}, \qquad \text{and} \qquad v(S) = - \frac{3}{2} \lvert S \rvert \lvert N \setminus S \rvert. 
\]
	
\begin{figure}[H]
\begin{center}
\includegraphics{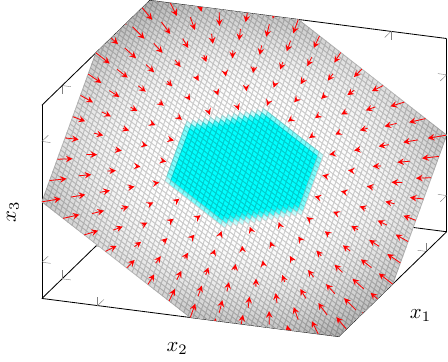}
\caption{Cohesion field for the game \(v(S) = -\frac{3}{2} \lvert S \rvert \lvert N \setminus S \rvert\).}
\end{center}
\end{figure}

Let also \(x = (-4, -3, 7)\) be as before. The vector associated with it is given by 
\[
\varphi(x) = \eta_a + 4 \eta_{ab} = \begin{pmatrix} 2 & 1 & -3 \end{pmatrix}^\top.  
\]
The largest and negative change is for the player \(c\), which could already be seen from the payment of \(c\) at the preimputation \(x\), as the game is symmetric. 
\end{example}

Let us make a few observations about the regularity of the cohesion flow. Both the maps \(e_S\) and \((\cdot )^+: t \mapsto t^+ = \max\{0, t\}\) are continuous, so \(\varphi\) is continuous as well. Moreover, the cohesion field is piecewise affine. Indeed, \(\varphi\) is affine on each region of \(X\) composed of preimputations with the same image through \(\phi\). We can rewrite 
\[
\varphi(x) = \sum_{S \in \phi(x)} e_S(x) \eta^S = \sum_{S \in \phi(x)} \eta^S \left( - \left(\eta^S\right)^\top x + v(S) \right) = - A x + b, 
\]
with \(A = \sum_{S \in \phi(x)} \eta^S \left( \eta^S \right)^\top\) and \(b = \sum_{S \in \phi(x)} v(S) \eta^S\). 

\medskip 

So far, we have defined two new maps, one returning real numbers, the other one returning preimputations, capturing different but related quantities: 
\begin{itemize}
\item the \emph{dissatisfaction field} \(\vartheta\), associating to any preimputation its level of dissatisfaction, and so instability;
\item the \emph{cohesion field} \(\varphi\), associating to any preimputation a vector, which direction and magnitude indicate the aggregated goal of the coalitions aggrieved by the current state of the game. 
\end{itemize}

In fact, we actually have that the dissatisfaction field is the opposite of the \emph{potential} of the cohesion field, i.e, we have 
\[
\varphi = - \nabla \vartheta. 
\]
So, the cohesion field associates with each preimputation the direction along which its instability decreases the most. Moreover, like the dissatisfaction field, \(\varphi\) characterizes the core, and provides new insights about his properties as a solution concept. 

\begin{proposition}\label{prop: core-invariant}
For any game \((N, v)\), we have \(C(v) = \{x \in X \mid \varphi(x) = 0\}\). 
\end{proposition}

\begin{proof}
In Appendix~\ref{proof: core-invariant}. 
\end{proof}

In particular, this results implies that, when the initial state of the game belongs to the core, it does not move according to the cohesion dynamics. Moreover, if one initial state evolves from outside the core onto its boundary, it becomes static. 

\medskip 

As we have previously discussed, the field \(\varphi\) associates to any preimputation \(x \in X\) the direction \(\varphi(x)\) along which the decrease of the dissatisfaction field \(\vartheta\) is the greatest. However, as soon as we move along that direction, the new optimal direction may not be the same. The evolution of the state of the game is given by the solution of a specific differential equation, which fully describes our dynamical system. 

\begin{definition}
The \emph{cohesion curve} of a given preimputation \(x \in X\) is a continuous curve \(\gamma_x: \mathbb{R} \to X\) which satisfies 
\[
\partial_t \gamma_x(t) = \varphi(\gamma_x(t)) \qquad \text{and} \qquad \gamma_x(0) = x. 
\]
The map \(\Phi: \mathbb{R} \times X \to X\) such that \(\Phi(t, x) = \gamma_x(t)\) is called the \emph{evolution map}. 
\end{definition}

For simplicity, we write \(\gamma_x = \bigcup_{t \in \mathbb{R}} \gamma_x(t)\), which we also call the cohesion curve. A cohesion curve of a preimputation \(x \in X\), if it exists, represents the evolution of the state of players' wealth driven by the dissatisfaction of the coalitions for which the payment at the current state \(\gamma_x(t)\) at time \(t\) is lower than their worth. It is the natural path followed by the state of the game during the renegociation in progress among all the players in \(N\), if their objective is to reduce the dissatisfaction \(\vartheta\). 

\begin{theorem}\label{th: unicity-cohesion-curve}
For any preimputation, there exists a unique cohesion curve. 
\end{theorem}

\begin{proof}
In Appendix~\ref{proof: unicity-cohesion-curve}. 
\end{proof}

A corollary from this result is that the cohesion curves do not `cross'. Hence, either two preimputations are on the same cohesion curve, i.e., one of them evolves into the other after some time, or they follow completely distinct paths. 

\medskip

For any state of the game represented by a preimputation, the cohesion field pushes the state in a unique direction, determined by the set \(\phi(x)\) of coalitions that are discontent with the current outcome. The strength of the flow is determined by the magnitude of the excesses endured by these coalitions. 

\medskip 

We are now interested in knowing where this flow leads, and whether there exists a preimputation or a set of preimputations towards which all cohesion curves converge, i.e., is there one or several attractor for this dynamical system. 

\medskip 

A natural candidate to be an attractor is the core of the game. The dissatisfaction field as well as the cohesion field vanish on the core. However, for now nothing prohibits the cohesion curve to cycle, and, while never stopping, could avoid the core forever. A similar issue happens in the study of the von Neumann-Morgenstern stability of the core, with the core containing all the undominated preimputations but failing to dominate the whole space. Indeed, the relation of domination being non-transitive, cycles of domination can occur. 

\medskip 

To investigate the cyclicality of the cohesion curves, we need to be able to differentiate along flow curves of a vector field. This generalizes the concept of directional derivative, in the sense that if the vector field \(f\) is constant, then the differentiation along \(f\) coincides with the derivative in the direction of \(u\).

\begin{definition}
Let \(g: X \to \mathbb{R}\) be a function, and let \(f: X \to X\) be a vector field. The \emph{Lie derivative} of \(g\) with respect to \(f\) is the scalar field
\[
\begin{tabular}{lccl}
    \( \mathcal{L}_f g: \) & \( X \) & \( \longrightarrow \) & \( \mathbb{R} \) \\[0.1cm]
    & \( x \) & \( \longmapsto \) & \( \displaystyle \mathcal{L}_f g(x) \coloneqq \langle \nabla_x g, f(x) \rangle \).
  \end{tabular} 
\]
\end{definition}

The Lie derivatives are useful to compute the evolution of certain quantities, measured by a function \(g\), when the state of the game evolves according to the field \(f\). The Lie derivative of \(g\) with respect to \(f\) at \(x\) is the derivative of \(g\) at \(x\) in the direction of \(f(x)\). Another powerful tool for the study of attractors are the Lyapunov functions. 

\begin{definition}[{Lyapunov~\cite{lyapunov1907probleme}}]\label{def: lyapunov}
	Let \(f\) be a vector field on \(X\), let \(Y\) be a closed invariant set of \(f\) and let \(E \supseteq Y\). A continuously differentiable function \(V: E \to \mathbb{R}\) is called a \emph{strict Lyapunov function} for \(f\) at \(Y\) is 
	\begin{enumerate}
	\item for each \(y \in Y\), we have \(V(y) = 0\) and for all \(x \in E \setminus Y\), we have \(V(x) > 0\) ;
	\item for each \(x \in E \setminus Y\), we have \(\mathcal{L}_f V(x) < 0\). 
	\end{enumerate}
\end{definition}

A Lyapunov function can be thought of as a function associating to each state the energy of a physical system studied with these equations. Naturally, the system evolves towards the states with the lowest energy, corresponding to the negativity of the Lie derivative of the Lyapunov function along the flow. 

\medskip 

In the social science context of this paper, the quantity represented by a Lyapunov function for the cohesion field \(\varphi\) at the core would indicate how much the current state fails at being in the core. Therefore, a very natural candidate as a Lyapunov function for \(\varphi\) at the core is the dissatisfaction field. 

\begin{lemma}\label{lemma: dissatisfaction-lyapunov}
For a balanced game, the dissatisfaction field is a strict Lyapunov function for the bargaining field at the core whose domain is the whole space of preimputations. 
\end{lemma}

\begin{proof}
Thanks to Proposition~\ref{prop: alt-def-core}, \(\vartheta\) already satisfies the first condition of Definition~\ref{def: lyapunov}. Let \(x \in X \setminus C(v)\) be a preimputation outside the core. We use the fact that the dissatisfaction field \(\vartheta\) is the opposite of the potential of \(\varphi\), to compute the Lie derivative of \(\vartheta\) at \(x\) along \(\varphi\), given by 
\[
\mathcal{L}_\varphi \vartheta(x) = \langle \nabla_x \vartheta, \varphi(x) \rangle = - \langle \varphi(x), \varphi(x) \rangle = - \lVert \varphi(x) \rVert^2 < 0. 
\]
Therefore, \(\vartheta\) is a Lyapunov function for \(\varphi\) at \(C(v)\), defined on the whole set \(X\). 
\end{proof}

Subsequently, the dissatisfaction of a given state \(x \in X\), or equivalently, the intensity of its instability, strictly decreases along the cohesion curve \(\gamma_x\). Therefore, there cannot exist any cycles in cohesion curves, which leads to the main result of the paper. 

\begin{theorem}\label{th: main-theorem}
For any balanced game, the core is the unique minimal global attractor of the cohesion flow. 
\end{theorem}

\begin{proof}
In Appendix~\ref{proof: main-theorem}. 
\end{proof}

We therefore have proved that, when the dynamics of the game is driven by the reduction of the dissatisfaction endured by the coalitions, the state of the game eventually reaches the core, from any possible initial condition. This is not the case when we consider a domination-based dynamics of the game. 


\section{Concluding remarks and perspectives}\label{sec: conclusion}

The dissatisfaction field defined in Section~\ref{sec: dissatisfaction} gives a numerical value to each possible state of the game, evaluating its level of instability. It turns out that the dissatisfaction field is compatible with the outvoting relation, in the sense that the dissatisfaction of a preimputation does not exceed the dissatisfaction of the preimputations it outvotes. Moreover, the dissatisfaction field characterizes the core when it is non-empty.

\medskip 

Together with this dissatisfaction field, we defined the cohesion field, which aggregates the interests of each aggrieved coalition into a unique vector for each state. The cohesion curve through a given state is precisely the trajectory followed by the state such that, at any time, the tangent to this trajectory is the direction defined by the cohesion field. The latter being the gradient of the dissatisfaction field, the trajectories defined by the cohesion curves are the most efficient paths to reduce the instability of the states. Furthermore, all the trajectories of any possible state converge to the core, and the core is the minimal set attracting all trajectories. 

\medskip 

A natural continuation to these investigations would be to study the cohesion field in the context of non-transferable utility games, where the set of preimputations is a more general manifold. Another interesting context is when the core is empty. Indeed, the definition of this vector does not rely on the non-emptiness of the core, and the potential attractors of this dynamical in this context could provide a convenient tool for the study of unbalanced games. 

\medskip 

Another extension of this work could be to look at different vector fields, for example \(\tilde{\varphi}\) defined from \(\varphi\) but distorted by some coefficients: 
\[
\tilde{\varphi}(x) = \sum_{S \in \phi(x)} a_S e_S(x) \eta^S. 
\]
In this setting, it could be possible to define a dynamical system converging towards the prenucleolus, where the coefficients \(a_S\) interact with the weights of balanced collections in a way resembling Kohlberg's criterion. Also, to have a closer connection to the von Neumann-Morgenstern stability, one could seek for a dynamical system with a vector field that always give an infinitesimal step that dominates the current state of the game, or that is determined by a bargaining process involving the frustrated coalitions. 

\medskip 

Finally, the cohesion flow could be converted into a gradient descent algorithm, to compute a core element of a general cooperative game without relying on some specific structure of the game. Such an algorithm could perhaps be used to check whether the core of a given game is non-empty.


\sloppy

\printbibliography{}


\appendix

\section{Proof of Proposition~\ref{prop: core-invariant}}\label{proof: core-invariant}

For this proof, we need results about \emph{unbalanced collections}, the dual counterparts of the balanced collections. 

\begin{definition}
	A collection of coalitions \(\mathcal{B} \subseteq \mathcal{N}\) is said to be \emph{balanced} if there exists a set of positive weights \(\{\lambda_S \mid S \in \mathcal{B}\}\) such that \(\sum_{S \in \mathcal{B}} \lambda_S \mathbf{1}^S = \mathbf{1}^N\). A collection \(\mathcal{C} \subseteq \mathcal{N}\) is said to be \emph{unbalanced} if it does not contain a balanced collection. 
\end{definition}

We can characterize geometrically the balanced collections using the vectors \(\eta^S\). 

\begin{proposition}\label{prop: balanced-zero}
	A collection of coalitions \(\mathcal{B} \subseteq \mathcal{N}\) is balanced if and only if there exists a set of positive numbers \(\{\theta_S \mid S \in \mathcal{B}\}\) such that \(\sum_{S \in \mathcal{B}} \theta_S \eta^S = 0_{\mathbb{R}^N}\). 
\end{proposition}

\begin{proof}
Let \(\mathcal{B}\) be a balanced collection. Then, there exists a set of positive numbers \(\{\lambda_S \mid S \in \mathcal{B}\}\) such that \(\sum_{S \in \mathcal{B}} \lambda_S \mathbf{1}^S = \mathbf{1}^N\). Denote by \(p\) the orthogonal projection from \(\mathbb{R}^N\) onto \(X\). Recall that \(p(\mathbf{1}^S) = \eta^S\). Then, projecting the equality above yields
\[
\sum_{S \in \mathcal{B}} \lambda_S \eta^S = p \left( \mathbf{1}^N \right) = 0_{\mathbb{R}^N}. 
\]
Conversely, if there exists positive numbers such that \(\sum_{S \in \mathcal{B}} \theta_S \eta^S = 0_{\mathbb{R}^N}\), then we get that \( \sum_{S \in \mathcal{B}} \theta_S \in \mathrm{ker}(p)\), i.e., there exists \(\alpha \in \mathbb{R}\) such that 
\[
\sum_{S \in \mathcal{B}} \theta_S = \alpha \mathbf{1}^N. 
\]
Since all \(\theta_S\) are positive, then \(\alpha > 0\), and the result follows. 
\end{proof}

The connection between the unbalanced collections and the dynamical systems we are studying comes from the following result. 

\begin{proposition}[{\cite[Lemma~8.1]{grabisch2021characterization}}]\label{prop: feasible-unbalanced}
	Let \((N, v)\) be a balanced game and \(x \not \in C(v)\). Then \(\phi(x)\) is a non-empty unbalanced collection. 
\end{proposition}

Thanks to this result, we can prove Proposition~\ref{prop: core-invariant}. 

\begin{proof}[Proof of Proposition~\ref{prop: core-invariant}]
When a preimputation \(x \in X\) is outside the core, we have \(\phi(x) \neq \emptyset\). By Proposition~\ref{prop: feasible-unbalanced}, we have that \(\phi(x)\) is unbalanced. By definition of \(( \cdot )^+ \) and Proposition~\ref{prop: balanced-zero}, we have that 
\[
\varphi(x) = \sum_{S \in \mathcal{N}} e_S(x)^+ \eta^S \neq 0, 
\]
and therefore \(C(v) \supseteq \{x \in X \mid \varphi(x) = 0\}\). However, if \(x\) belongs to the core, the sum above is empty, then, by double inclusion, \(C(v) = \{x \in X \mid \varphi(x) = 0\}\). 
\end{proof}


\section{Proof of Theorem~\ref{th: unicity-cohesion-curve}}\label{proof: unicity-cohesion-curve}

To prove this theorem, we use the global Cauchy-Lipschitz theorem. 

\begin{definition}
Let \(E\) be an open set of \(\mathbb{R}^n\). We say that the function \(f: E \to \mathbb{R}^n\) is \emph{\(k\)-Lipschitz continuous} on a subset \(W\) of \(E\) if there exists a constant \(k\) such that, for all \((x_1, x_2) \in W^2\), we have 
\[
\lVert f(x_1) - f(x_2) \rVert \leq k \lVert x_1 - x_2 \rVert. 
\]
\end{definition}

\begin{theorem}[Global Cauchy-Lipschitz Theorem for autonomous vector field]\label{th: cauchy-lipschitz}
Let \(d\) be a positive integer, and let \(f: \mathbb{R}^d \to \mathbb{R}^N\) be a Lipschitz continuous vector field on \(\mathbb{R}^d\). Then, for all \(y \in \mathbb{R}^d\), the initial value problem 
\[
\partial_t c_y(t) = f(c_y(t)), \qquad c_y(0) = y, 
\]
has a unique solution \(c_y\) defined on \(\mathbb{R}\) which is continuously differentiable. 
\end{theorem}

\begin{proof}[Proof of Theorem~\ref{th: unicity-cohesion-curve}]
 In order to apply Theorem~\ref{th: cauchy-lipschitz} to the cohesion flow \(\varphi\), we need to check that it is Lipschitz continuous on the whole linear subspace of preimputations. Let us have a closer look at the map \(\phi: x \mapsto 2^\mathcal{N}\). Its codomain is finite, so there are only a finite number of preimages of \(\phi\). On each of these preimages, the cohesion field is affine, hence, have a constant gradient. Taking the maximum of the norms of the gradients on each preimage of \(\phi\) gives the Lipschitz constant of \(\varphi\), which is necessarily finite as \(2^\mathcal{N}\) is finite. The application of Theorem~\ref{th: cauchy-lipschitz} on \(\varphi\) finishes the proof. 
\end{proof}


\section{Proof of Theorem~\ref{th: main-theorem}}\label{proof: main-theorem}

For this proof, we need some technical properties about the dissatisfaction field. 

\begin{definition}[{\cite[Definition~3.4]{bhatia1967asymptotic}}]
	A continuous scalar function \(V\) defined on a set \(N \subseteq X\) is said to be \emph{uniformly unbounded} on \(N\) if given any \(\alpha > 0\), there exists a compact set \(K \subsetneq N\), such that \(V(x) \geq \alpha\) for \(x \not \in K\). 
\end{definition}

\begin{lemma}\label{lemma: uniformly-unbounded}
	The dissatisfaction field is uniformly unbounded. 
\end{lemma}

\begin{proof}[Proof of Lemma~\ref{lemma: uniformly-unbounded}]
	We prove this result by choosing the compact sets \(K\) to be \(\varepsilon\)-cores. For any real number \(\varepsilon\), the \(\varepsilon\)-core \cite[Definition~4.1]{driessen2013cooperative} of a game \((N, v)\) is defined by 
\[
C_\varepsilon(v) = \{x \in X \mid x(S) \geq v(S) - \varepsilon, \text{ for all } S \in \mathcal{N} \setminus \{N\}\}. 
\]
For any choice of \(\varepsilon\), the \(\varepsilon\)-cores are compact, convex subsets of \(X\)~\cite[p.~22]{driessen2013cooperative}. For each preimputation \(x \not \in C_\varepsilon(v)\), there exists a coalition \(S\) such that \(e_S(x) > \varepsilon\), hence \(S \in \phi(x)\) and \( \vartheta(x) > \frac{1}{2} \varepsilon^2\). Consequently, for any \(\alpha > 0\), if \(x \not \in C_{\sqrt{2\alpha}}\), then \(\vartheta(x) > \alpha\), therefore \(\vartheta\) is uniformly unbounded. 
\end{proof}

Uniform unboundedness allows us to identify global attractor as follows.  

\begin{theorem}[{\cite[Theorem~3.10]{bhatia1967asymptotic}}]\label{th: characterization-global-attractor}
	A compact set \(M \subseteq X\) is a global attractor if and only if there exists a continuous uniformly unbounded strict Lyapunov function on \(X\). 
\end{theorem}

We now have all the tools to prove Theorem~\ref{th: main-theorem}. 

\begin{proof}[Proof of Theorem~\ref{th: main-theorem}]
	By definition, the core is closed and by Proposition~\ref{prop: core-invariant}, the core is invariant. By Lemma~\ref{lemma: dissatisfaction-lyapunov}, we have that the dissatisfaction field is a strict Lyapunov function for \(\varphi\) at \(C(v)\) and by Lemma~\ref{lemma: uniformly-unbounded}, we have that the dissatisfaction field is uniformly unbounded. By Theorem~\ref{th: characterization-global-attractor}, the core is a global attractor of the cohesion flow. Moreover, because the cohesion field vanishes on the core, no subset of it can be an attractor. Furthermore, because the cohesion field is never zero outside of the core, each attractor must contain the core. So, the core is the unique minimal global attractor of the cohesion flow. 
\end{proof}

\end{document}